\documentclass[letterpaper,11pt]{article}
\pdfoutput=1

\usepackage[utf8]{inputenc}

\usepackage{jheppub}
\usepackage{multirow, amsthm}
\usepackage{braket}
\usepackage{mathrsfs}
\usepackage{tikz}
\usepackage[caption=false]{subfig}
\usepackage{xspace}
\usepackage{xcolor}
\usepackage{float,color}
\usepackage{siunitx}
\usepackage[countmax]{subfloat}

\newtheorem{theorem}{Theorem}
\providecommand{\href}[2]{#2}
\usepackage{comment}

\definecolor{darkred}{rgb}{0.5,0.0,0.0}
\definecolor{darkblue}{rgb}{0.0,0.0,0.9}
\definecolor{darkerblue}{rgb}{0.0,0.0,0.5}
\definecolor{darkgreen}{rgb}{0.0,0.5,0.0}
\definecolor{black}{rgb}{0.0,0.0,0.0}
\definecolor{brown}{rgb}{0.6,0.4,0.2}

\DeclareRobustCommand{\Sec}[1]{Sec.~\ref{#1}}
\DeclareRobustCommand{\Secs}[2]{Secs.~\ref{#1} and \ref{#2}}

\DeclareRobustCommand{\Fig}[1]{Fig.~\ref{#1}}

\DeclareRobustCommand{\Eq}[1]{Eq.~(\ref{#1})}
\DeclareRobustCommand{\Eqs}[2]{Eqs.~(\ref{#1}) and (\ref{#2})}
\DeclareRobustCommand{\Ref}[1]{Ref.~\cite{#1}}
\DeclareRobustCommand{\Refs}[1]{Refs.~\cite{#1}}

\title{\boldmath Classification without labels: \\ Learning from mixed samples in high energy physics}

\preprint{ 
\begin{flushright}
MIT--CTP 4922
 \end{flushright}}

\author[a]{Eric M. Metodiev,}
\author[b]{Benjamin Nachman,}
\author[a]{and Jesse Thaler}

\affiliation[a]{Center for Theoretical Physics, Massachusetts Institute of Technology, Cambridge, MA 02139, USA}
\affiliation[b]{Physics Division, Lawrence Berkeley National Laboratory, Berkeley, CA 94720, USA}

\emailAdd{metodiev@mit.edu}
\emailAdd{bpnachman@lbl.gov}
\emailAdd{jthaler@mit.edu}

\abstract{
Modern machine learning techniques can be used to construct powerful models for difficult collider physics problems.  In many applications, however, these models are trained on imperfect simulations due to a lack of truth-level information in the data, which risks the model learning artifacts of the simulation.  In this paper, we introduce the paradigm of classification without labels (CWoLa) in which a classifier is trained to distinguish statistical mixtures of classes, which are common in collider physics.  Crucially, neither individual labels nor class proportions are required, yet we prove that the optimal classifier in the CWoLa paradigm is also the optimal classifier in the traditional fully-supervised case where all label information is available.  After demonstrating the power of this method in an analytical toy example, we consider a realistic benchmark for collider physics:  distinguishing quark- versus gluon-initiated jets using mixed quark/gluon training samples.  More generally, CWoLa can be applied to any classification problem where labels or class proportions are unknown or simulations are unreliable, but statistical mixtures of the classes are available.
}

\begin{document} 
\maketitle
\flushbottom

\section{Introduction}
\label{sec:intro}

In the data-rich environment of the Large Hadron Collider (LHC), machine learning techniques have the potential to significantly improve on many classification, regression, and generation problems in collider physics.  There has been a recent surge of interest in applying deep learning and other modern algorithms to a wide variety of problems, such as jet tagging~\cite{Cogan:2014oua,Almeida:2015jua,deOliveira:2015xxd,Baldi:2016fql,Barnard:2016qma,Kasieczka:2017nvn,Komiske:2016rsd,deOliveira:2017pjk,Komiske:2017ubm,Louppe:2017ipp,ATLAS-CONF-2017-064,ATL-PHYS-PUB-2017-013,ATL-PHYS-PUB-2017-004,CMS-DP-2017-005,CMS-DP-2017-013,ATL-PHYS-PUB-2017-003,Butter:2017cot,Pearkes:2017hku,Datta:2017rhs,ATL-PHYS-PUB-2017-017,CMS-DP-2017-027}. Despite the power of these methods, they all currently rely on significant input from simulations.  Existing multivariate approaches for classification used by the LHC experiments all have some degree of mis-modeling by simulations and must be corrected post-hoc using data-driven techniques~\cite{Aad:2015ydr,Chatrchyan:2012jua,Aad:2014gea,CMS:2013kfa,CMS-DP-2016-070,Aad:2015rpa,Khachatryan:2014vla,Aad:2016pux,CMS:2014fya}.  The existence of these \textit{scale factors} is an indication that the algorithms trained on simulation are sub-optimal when tested on data.  Adversarial approaches can be used to mitigate potential mis-modeling effects during training at the cost of algorithmic performance~\cite{Louppe:2016ylz}.  The only solution that does not compromise performance is to train directly on data.  This is often thought to not be possible because data is unlabeled.

In this paper, we introduce \textit{classification without labels} (CWoLa, pronounced ``koala''), a paradigm which allows robust classifiers to be trained directly on data in scenarios common in collider physics.  Remarkably, the CWoLa method amounts to only a minor variation on well-known machine learning techniques, as one can effectively utilize standard fully-supervised techniques on two mixed samples.  As long as the two samples have different compositions of the true classes (even if the label proportions are unknown), we prove that the optimal classifier in the CWoLa framework is the optimal classifier in the fully-supervised case.\footnote{After we developed this framework, we learned of a mathematically equivalent (but conceptually different) rephrasing of CWoLa in the language of learning from random noisy labels in \Ref{scott2013}, where a version of Theorem~\ref{th:optXYAB} also appears. See the discussion in \Sec{sec:cwola}.} 
In practice, after training the classifier on large event samples without using label information, the operating points of the classifier can be determined from a small sample where at least the label proportions are known.

The CWoLa paradigm is part of a broader set of classification frameworks that fall under the umbrella of \emph{weak supervision}.  These frameworks go beyond the standard fully-supervised paradigm with the goal of learning from partial, non-standard, or imperfect label information. See~\Ref{hernandez2015} for a recent review and comprehensive taxonomy.  Weak supervision was first applied in the context of high energy physics in~\Ref{Dery:2017fap} to distinguish jets originating from quarks from those originating from gluons using only class proportions during training; this paradigm is known as \textit{learning from 
label proportions} (LLP)~\cite{quadrianto2009,patrini2014almost}.  For quark versus gluon jet tagging, LLP was an important development because useful quark/gluon discrimination information is often subtle and sensitive to low-energy or wide-angle radiation inside jets, which may not be modeled correctly in parton shower generators~\cite{Gras:2017jty}.  The main drawback of LLP, however, is that there is still uncertainty in the quark/gluon labels themselves, since quark/gluon fractions are determined by matrix element calculations convolved with parton distribution functions, which carry their own uncertainties.  The CWoLa paradigm sidesteps the issue of quark/gluon fractions entirely, and only relies on the assumption that the samples used for training are proper mixed samples without contamination or sample-dependent labeling.

The ideas presented below may prove useful for a wide variety of machine learning applications, but for concreteness we focus on \textit{classification}.  It is worth emphasizing that the CWoLa framework can be applied to a huge variety of classifiers\footnote{CWoLA can be applied to train any classifier with a threshold that can be varied to sweep over operating points. $k$-nearest neighbors classification, for instance, does not have this property.} without modification to the training procedure, by simply training on mixed event samples instead of on pure samples.   By contrast, LLP-style weak supervision such as in \Ref{Dery:2017fap} requires a non-trivial modification to the loss function.\footnote{The recent study in \Ref{Cohen:2017exh}, which was initially inspired by the LLP paradigm, is actually performing weak supervision using the CWoLa approach.  We thank Timothy Cohen, Marat Freytsis, and Bryan Ostdiek for clarifications on this point.}  For this reason, CWoLa can be applied even for classifiers that are not trained in terms of loss functions at all.

Despite the power and simplicity of the CWoLa approach, there are some important limitations to keep in mind.  First, the optimality of CWoLa is only true asymptotically; for a finite training set and a realistic machine learning algorithm, there can be differences, as discussed more below.  Second, CWoLa does not apply when one class does not already exist in the data, as may be the case in a search for physics beyond the Standard Model (SM) with an exotic signature.  That said, if the new physics can be decomposed into SM-like components, such as different types of jets, then CWoLa may once again be possible.  Third, when the CWoLa strategy is employed for training in one event topology and testing in another event topology, there may be systematic uncertainties associated with the extrapolation.  Of course, this is also true for traditional fully-supervised classification, which may introduce residual dependence on simulation; indeed, one could even combine adversarial approaches with CWoLa in this case to mitigate simulation dependence~\cite{Louppe:2016ylz}.  Finally, the CWoLa approach presented here only applies to mixtures of two categories, and further developments would be needed to disentangle multicategory samples.

The remainder of this paper is organized as follows.  In \Sec{sec:withandwithoutlabels}, we explain the theoretical foundations of the CWoLa paradigm and contrast it with LLP-style weak supervision and full supervision.  We illustrate the power of CWoLa with a toy example of two gaussian random variables in \Sec{sec:toy}.  We then apply CWoLa to the challenge of quark versus gluon jet tagging in \Sec{sec:qg}, using a dense network of five standard quark/gluon discriminants to highlight the performance of CWoLa on mixed samples. The paper concludes in \Sec{sec:conc} with a summary and future outlook.

\section{Machine learning with and without labels}
\label{sec:withandwithoutlabels}

The goal of classification is to distinguish two processes from each other: signal $S$ and background $B$.  Let $\vec{x}$ be a list of observables that are useful for distinguishing signal from background, and define $p_S(\vec{x})$ and $p_B(\vec{x})$ to be the probability distributions of $\vec{x}$ for the signal and background, respectively.  A classifier $h:\vec{x}\mapsto \mathbb{R}$ is designed such that higher values of $h$ are more signal-like and lower values are more background-like.  A classifier operating point is defined by a threshold cut $h>c$; the signal efficiency is then $\epsilon_S=\int \mathrm{d}\vec{x} \, p_S(\vec{x}) \, \Theta(h(\vec{x})-c)$ and the background efficiency (i.e.~mistag rate) is $\epsilon_B=\int \mathrm{d}\vec{x} \, p_B(\vec{x})\, \Theta(h(\vec{x})-c)$, for the Heaviside step function $\Theta$.  The performance of a classifier $h$ can be described by its receiver operating characteristic (ROC) curve which is the function $1 - \epsilon_B^h(\epsilon_S)$.  A classifier $h$ is \textit{optimal} if for any other classifier $h'$, $\epsilon_B^{h'}(\epsilon_S)\geq \epsilon_B^h(\epsilon_S)$ for all possible $\epsilon_S$.  By the Neyman-Pearson lemma~\cite{nplemma}, an optimal classifier is the likelihood ratio: $h_\text{optimal}(\vec{x}) = p_S(\vec{x})/p_B(\vec{x})$.  Therefore, the goal of classification is to learn $h_\text{optimal}$ or any classifier that is monotonically related to it.  

In practice, one learns to approximate $h_\text{optimal}(\vec{x})$ from a set of signal and background $\vec{x}$ examples (\textit{training data}).  When the dimensionality of $\vec{x}$ is small and the number of examples large, it is often possible to approximate $p_S(\vec{x})$ and $p_B(\vec{x})$ directly by using histograms.  When the dimensionality is large, an explicit construction is often not possible.  In this case, one constructs a loss function that is minimized using a machine learning algorithm like a boosted decision tree or (deep) neural network.  The following section describes three paradigms for learning $h_\text{optimal}(\vec{x})$ with different amounts of information available at training time:  full supervision, LLP, and CWoLa.  The ideas presented here apply to any procedure for constructing $h_\text{optimal}(\vec{x})$.

\subsection{Full supervision}
\label{sec:full}

Fully supervised learning is the standard classification paradigm.  Each example $\vec{x}_i$ comes with a label $u_i\in\{S,B\}$.  For models trained to minimize loss functions, typical loss functions are the mean squared error:
\begin{equation}
\ell_{\rm MSE}=\frac{1}{N}\sum_{i=1}^N \Big(h(\vec{x}_i)-\mathbb{I}(u_i=S) \Big)^2,
\end{equation}
for the indicator function $\mathbb{I}$, or the cross-entropy:
\begin{equation}
\ell_{\rm CE} = -\frac1N\sum_{i=1}^N \Big( \mathbb{I}(u_i=S) \log h(\vec{x}_i)+ \big(1 - \mathbb{I}(u_i=S) \big) \log \big(1 - h(\vec{x}_i)\big) \Big),
\end{equation}
where $N$ is the size of the subset (\textit{batch}) of the available training data.  With large enough training samples, flexible enough model parameterization, and suitable minimization procedure, the learned $h$ should approach the performance of $h_\text{optimal}$.

\subsection{Learning from label proportions}
\label{sec:weak}

For weak supervision, one does not have complete and/or accurate label information.  Here, we consider the case of accurate labels, but in the context of mixed samples.  Consider two processes $M_1$ and $M_2$ that are mixtures of the original signal and background processes:
\begin{align}
p_{M_1}(\vec{x})&=f_1\, p_S(\vec{x})+(1-f_1) \, p_B(\vec{x})\label{eq:pA},\\
p_{M_2}(\vec{x})&=f_2 \, p_S(\vec{x})+(1-f_2)\, p_B(\vec{x})\label{eq:pB},
\end{align}
with the signal fractions satisfying $0\leq f_2 < f_1 \leq 1$.

Instead of having training data labeled as being from $p_S$ or $p_B$, we are now only given examples drawn from $p_{M_1}$ and $p_{M_2}$ with the corresponding $M_1$ and $M_2$ labels.  We are however told $f_1$ and $f_2$ ahead of time.  The resulting optimization problems are much less constrained than those in \Sec{sec:full}, but learning is still possible.  The key is to use several different mixed samples with sufficiently different fractions in order to avoid trivial failure modes, as discussed in~\Ref{Dery:2017fap}.  One possible loss function is given by:
\begin{equation}\ell_{\rm LLP} = \left| \sum_{i=1}^{N_{M_1}}\frac{h(\vec{x}_i)}{N_{M_1}} -f_1 \right|+ \left| \sum_{j=1}^{N_{M_2}}\frac{h(\vec{x}_j)}{N_{M_2}} -f_2 \right|,
\end{equation}
where $N_{M_1}$ and $N_{M_2}$ are the number of $M_1$ and $M_2$ examples in the batch.  One could extend (and improve) this paradigm by adding in more samples with different fractions, but we consider only two here for simplicity.

\subsection{Classification without labels}
\label{sec:cwola}

CWoLa is an alternative strategy for weak supervision in the context of mixed samples.  Rather than modifying the loss function to accommodate the limited information as in \Sec{sec:weak}, the CWoLa approach is to simply train the model to discriminate the mixed samples $M_1$ and $M_2$ from one another.  The classifier $h$ trained to distinguish $M_1$ from $M_2$ (using full supervision) is then directly applied to distinguish $S$ from $B$. An illustration of this technique is shown in \Fig{fig:framework}. Remarkably, this procedure results in an optimal classifier (as defined in the beginning of Sec.~\ref{sec:withandwithoutlabels}) for the $S$ versus $B$ classification problem:
\begin{theorem}
\label{th:optXYAB}
Given mixed samples $M_1$ and $M_2$ defined in terms of pure samples $S$ and $B$ using \Eqs{eq:pA}{eq:pB} with signal fractions $f_1 > f_2$, an optimal classifier trained to distinguish $M_1$ from $M_2$ is also optimal for distinguishing $S$ from $B$.
\end{theorem}
\begin{proof}
The optimal classifier to distinguish examples drawn from $p_{M_1}$ and $p_{M_2}$ is the likelihood ratio $L_{M_1/M_2}(\vec{x}) = p_{M_1}(\vec{x})/p_{M_2}(\vec{x})$.  Similarly, the optimal classifier to distinguish examples drawn from $p_S$ and $p_B$ is the likelihood ratio $L_{S/B}(\vec{x})=p_S(\vec{x})/p_B(\vec{x})$.  Where $p_B$ has support, we can relate these two likelihood ratios algebraically:
\begin{equation}
L_{M_1/M_2} = \frac{p_{M_1}}{p_{M_2}}=\frac{f_1 \, p_S + (1-f_1) \, p_B}{f_2 \,  p_S + (1-f_2) \, p_B} = \frac{f_1 \, L_{S/B} + (1-f_1) }{f_2 \, L_{S/B} + (1-f_2)},
\end{equation}
which is a monotonically increasing rescaling of the likelihood $L_{S/B}$ as long as $f_1 >  f_2$, since $\partial_{L_{S/B}}L_{M_1/M_2}=(f_1-f_2)/(f_2L_{S/B}-f_2+1)^2>0$. If $f_1 < f_2$, then one obtains the reversed classifier. Therefore, $L_{S/B}$ and $L_{M_1/M_2}$ define the same classifier.
\end{proof}

\begin{figure}[t]
\centering
\includegraphics[scale = 0.7]{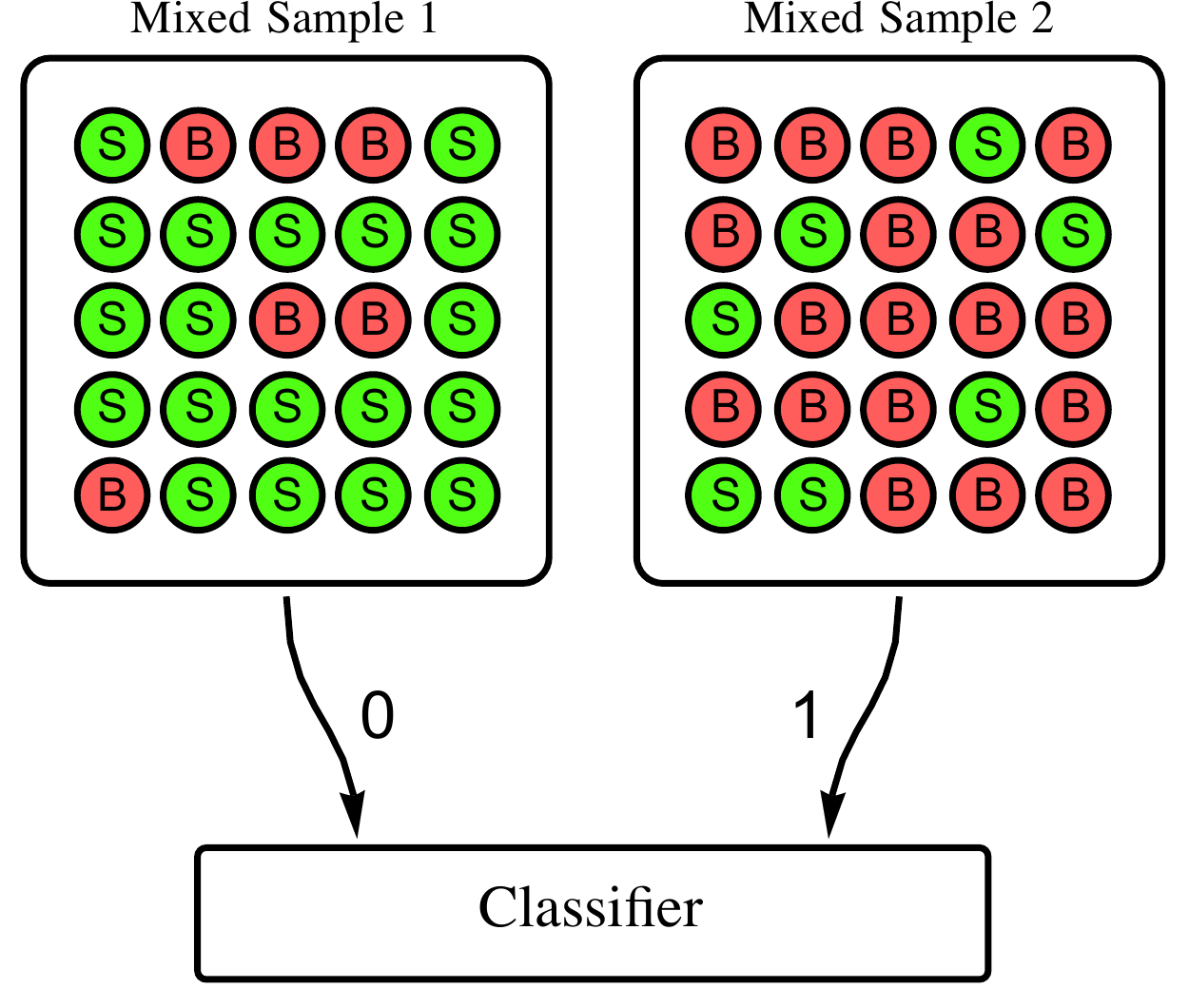}
\caption{\label{fig:framework}
An illustration of the CWoLa framework. Rather than being trained to directly classify signal ($S$) from background ($B$), the classifier is trained by standard techniques to distinguish data as coming either from the first or second mixed sample, labeled as 0 and 1 respectively. No information about the signal/background labels or class proportions in the mixed samples is used during training.}
\end{figure}

An important feature of CWoLa is that, unlike the LLP-style weak supervision in \Sec{sec:weak}, the label proportions $f_1$ and $f_2$ are not required for training.  Of course, this proof only guarantees that the optimal classifier from CWoLa is the same as the optimal classifier from fully-supervised learning.  We explore the practical performance of CWoLa in \Secs{sec:toy}{sec:qg}.

The problem of learning from unknown mixed samples can be shown to be mathematically equivalent to the problem of learning with asymmetric random label noise, where there have been recent advances~\cite{scott2013,natarajan2013learning}. The equivalence of these frameworks follows from the fact that randomly flipping the labels of pure samples, possibly with different flip probabilities for signal and background, produces mixed samples. In the language of noisy labels, \Ref{scott2013} argues that even unknown class proportions can be estimated from mixed samples under certain conditions using mixture proportion estimation~\cite{scott2015rate}, which may have interesting applications in collider physics.  There are also connections between learning from unknown mixed samples and the \emph{calibrated classifiers} approach in \Ref{Cranmer:2015bka}, where measurement of the class proportions from unknown mixtures is also shown to be possible.

\subsection{Operating points}
\label{sec:OP}

While the optimal classifier from CWoLa is independent of the mixed sample compositions, some minimal input is needed in order to establish classification operating points.  Specifically, to define a cut on the classifier $h$ at a value $c$ to achieve signal efficiency $\epsilon_S$, one requires some degree of label information.

One practical strategy is to use CWoLa to train on two large mixed samples without label or class proportion information, and then benchmark it on two smaller samples where the class proportions $f_1$ and $f_2$ are precisely known.  In that case, one can solve a simple system of equations on the smaller samples:
\begin{align}
\Pr(h(x)>c\, |\, M_1) &= \epsilon_S \, f_1  + \epsilon_B\, (1-f_1) \label{eq:roc1}\\
\Pr(h(x)>c\, |\, M_2) &= \epsilon_S \, f_2  + \epsilon_B\, (1-f_2) \label{eq:roc2},
\end{align}
where the probabilities can be estimated numerically by counting the number of events that pass the classifier cut in some sample, e.g.\ $\Pr(h(x)>c\,|\, M_1)\approx \sum_{x\in {\mathcal{M}_1}} \mathbb{I}[h(x) > c]/|\mathcal{M}_1|$, where $\mathcal M_1$ is the mixed sample data.   Thus with class proportions only, the ROC curve of a classifier can be determined.\footnote{We are grateful to Francesco Rubbo for bringing this to our attention.}

For the purpose of establishing working points, one might need to rely on simulations to determine the label proportions of the test samples.  In many cases, though, label proportions are better known than the details of the observables used to train the classifier.  For instance, in jet tagging, the label proportions of kinematically-selected samples are largely determined by the hard scattering process, with only mild sensitivity to effects such as shower mismodeling.  In this way, one is sensitive only to simulation uncertainties associated with sample composition, which in most cases are largely uncorrelated with uncertainties associated with tagging performance.

To summarize, the CWoLa paradigm does not need class proportions during training, and it only requires a small sample of test data where class proportions are known in order to determine the classifier performance and operating points, with minimal input from simulation.

\section{Illustrative example: Two gaussian random variables}
\label{sec:toy}

Before demonstrating the combination of CWoLa with a modern neural network, we first illustrate the various forms of learning discussed in \Sec{sec:withandwithoutlabels} through a simplified example where the optimal classifier can be obtained analytically.  Consider a single observable $x$ for distinguishing a signal $S$ from a background $B$.  For simplicity, suppose that the probability distribution of $x$ is a Gaussian with mean $\mu_S$ and standard deviation $\sigma_S$ for the signal and a Gaussian with mean $\mu_B$ and standard deviation $\sigma_B$ for the background.  We then consider the mixed samples $M_1$ and $M_2$ from \Eqs{eq:pA}{eq:pB} with signal fractions $f_1$ and $f_2$.

In this one-dimensional case, the optimal fully-supervised classifier can be constructed analytically:
\begin{align}
\label{eq:optimalclassifier}
h_\text{optimal}(x) = \frac{p_S(x)}{p_B(x)}.
\end{align}
Of course, non-parameterically estimating \Eq{eq:optimalclassifier} numerically requires a choice of binning which can introduce numerical fluctuations.  To avoid this effect, we discretize $x$ into $50$ bins between $-40$ and $40$ (under/overflow is added to the first/last bins).  There are then a finite number of possibilities for the likelihood ratio in \Eq{eq:optimalclassifier}.

\begin{figure}[p] 
\centering
\subfloat[]{
\includegraphics[width=0.48\textwidth]{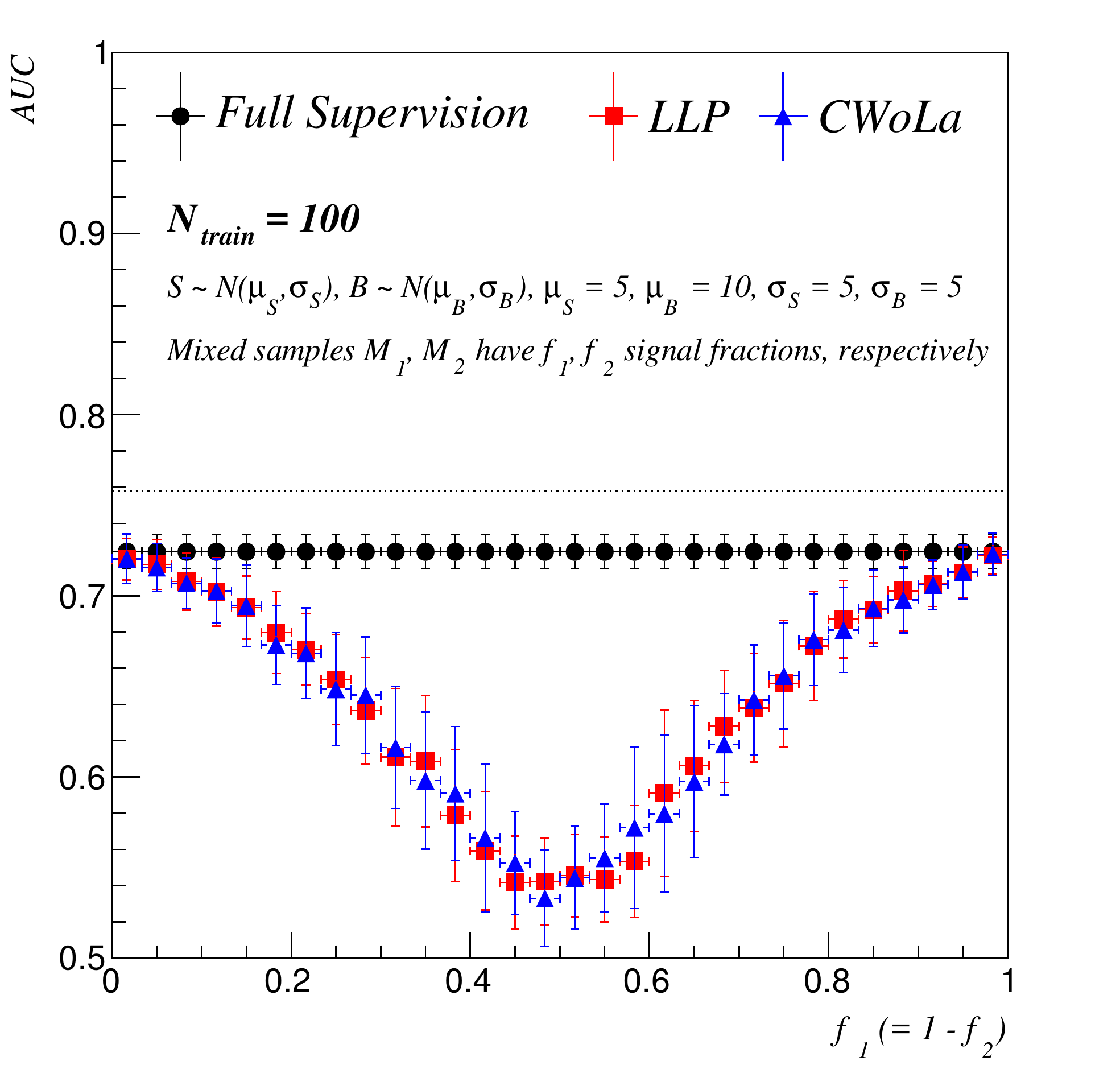}
}
\subfloat[]{
\includegraphics[width=0.48\textwidth]{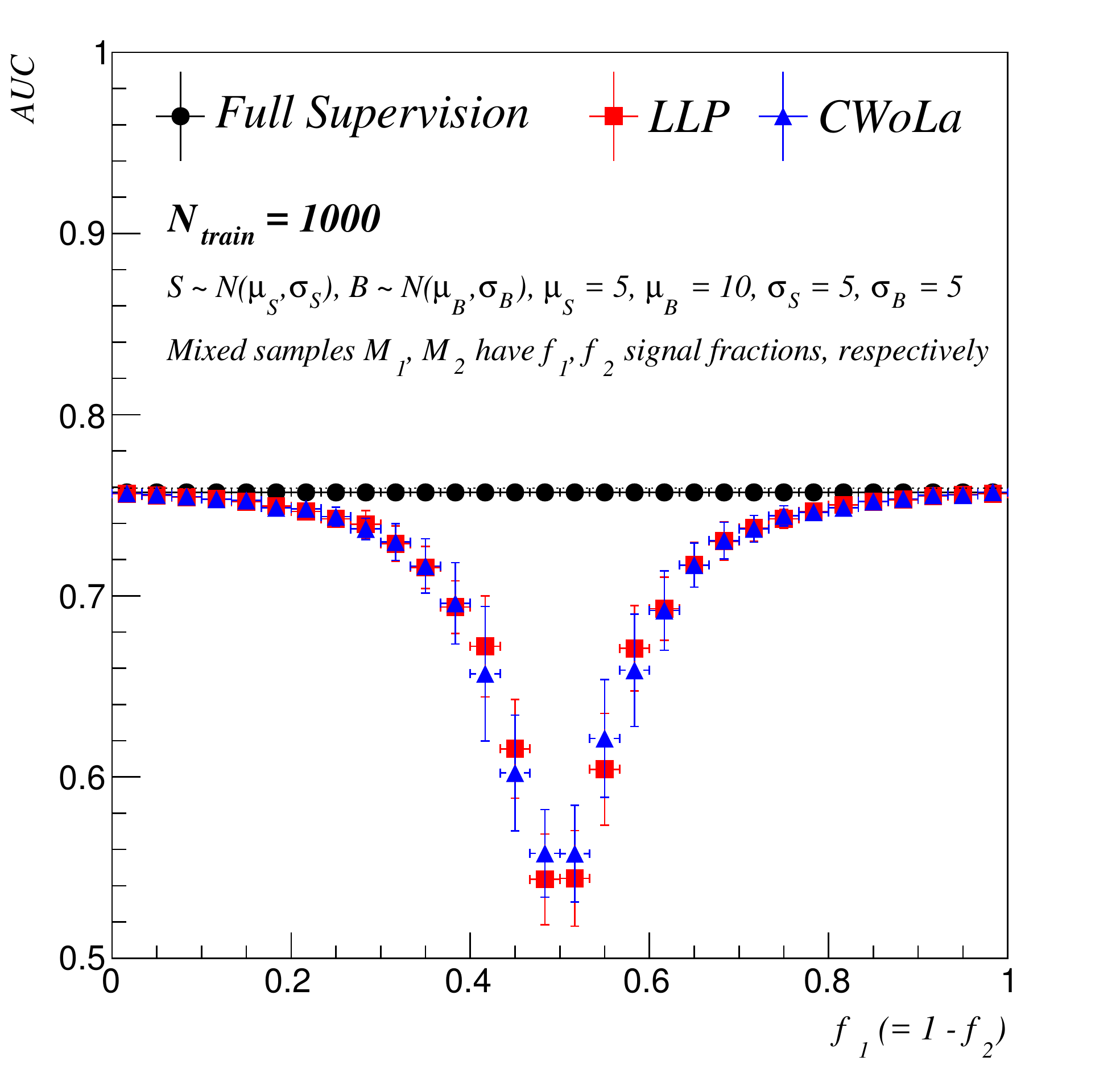}
}
\\
\subfloat[]{
\includegraphics[width=0.48\textwidth]{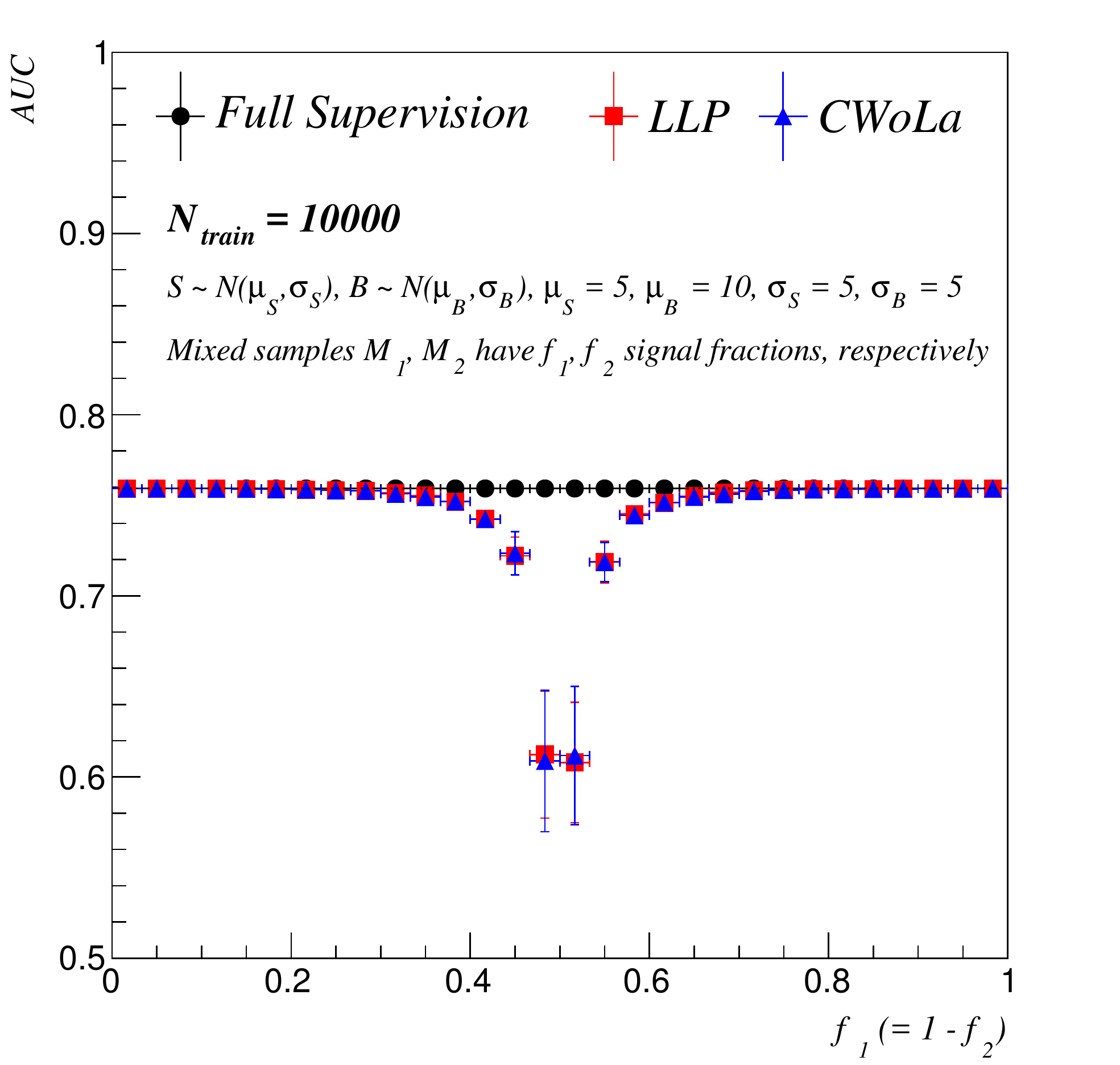}
}
\caption{The AUC for the LLP and CWoLa methods as a function of the signal fraction $f_1$, for training sizes $N_{\rm train}$ of (a) 100 events, (b) 1k events, and (c) 10k events.  Here, the complementary signal fraction is $f_2 = 1-f_1$.  By construction, the AUC for full supervision is independent of $f_1$.  The horizontal dashed line indicates the fully-supervised AUC with infinite training statistics. For $N_{\rm train}$ sufficiently large and $f_1$ sufficient far from $0.5$, all three methods converge to the optimal case.}
\label{fig:toy}
\end{figure}

Using a calligraphic font to denote explicit training samples, we test the following classifiers on signal ($\mathcal{S}$), background ($\mathcal{B}$), and mixed ($\mathcal{M}_{1,2}$) training samples of the same size:
\begin{enumerate}
\item {\bf Full Supervision} (\Sec{sec:full}): By construction, every example in the signal training dataset $\mathcal{S}$ is a signal event and every example in the background training set $\mathcal{B}$ is a background event.  The classifier is the numerical approximation to \Eq{eq:optimalclassifier}: 
\begin{align}
\label{eq:fullapprox}
h_\text{full}(x)=\frac{\sum_{y\in \mathcal{S}} \mathbb{I}[y=x]}{ \sum_{y\in \mathcal{B}} \mathbb{I}[y=x]}.
\end{align}

\item {\bf LLP}  (\Sec{sec:weak}): The events in the mixed training samples $\mathcal{M}_1$ and $\mathcal{M}_2$ are a mixture of signal and background events.  Weak supervision proceeds by solving the system of equations in \Eqs{eq:pA}{eq:pB} and using numerical estimates for $p_{M_1}$ and $p_{M_2}$:
\begin{align}
\label{eq:weak}
h_\text{LLP}(x)=\frac{(1-f_2)\sum_{y\in \mathcal{M}_1} \mathbb{I}[y=x]-(1-f_1)\sum_{y\in \mathcal{M}_2} \mathbb{I}[y=x]}{f_1\sum_{y\in \mathcal{M}_2} \mathbb{I}[y=x]-f_2\sum_{y\in \mathcal{M}_1} \mathbb{I}[y=x]}.
\end{align}

\item {\bf CWoLa} (\Sec{sec:cwola}): The input is the same as for the LLP case, though the fractions $f_1$ and $f_2$ are not needed as input. The CWoLa classifier is the same as in \Eq{eq:fullapprox}, only now signal and background distributions are replaced by the available mixed examples:
\begin{align}
\label{eq:fullapprox_cwol}
h_\text{CWoLa}(x)=\frac{\sum_{y\in \mathcal{M}_1} \mathbb{I}[y=x]}{ \sum_{y\in \mathcal{M}_2} \mathbb{I}[y=x]}.
\end{align}
\end{enumerate}
The performance of the classifiers trained in this way is evaluated on a holdout set of signal and background examples that is large enough such that statistical fluctuations are negligible.  We use the area under the curve (AUC) metric to quantify performance.  For continuous random variables, the AUC can be defined as $\Pr(h(x|S) > h(x|B))$.  This notion extends well to discrete random variables (indexed by integers):
\begin{align}
\text{AUC}=\sum_{i=1}\sum_{j=i+1}\Pr(x=i \,|\, S)\, \Pr(x=j\,| \, B)+\frac12\sum_{i=1}\Pr(x=i\, |\, S)\, \Pr(x=i\,| \,B).
\end{align}
For a properly constructed classifier, the AUC $\geq 0.5$.  In all of the numerical examples shown below, the classifier is inverted if necessary so that by construction, AUC $\geq 0.5$.

In \Fig{fig:toy}, we illustrate the performance of the three classification paradigms described above with 100, 1k, and 10k training examples each of $S$ and $B$, or $M_1$ and $M_2$ in the LLP and CWoLa cases, taking $f_1 = 1- f_2$ for concreteness.  Testing is performed on 100k $S$ and $B$ examples in all cases.  The LLP and CWoLa paradigms have nearly the same dependence on the number of training events and the signal fraction $f_1$.  The full supervision does not depend on the signal composition of $M_1$ and $M_2$ as it is trained directly on labeled signal and background examples. As expected, the performance is poor when the number of training examples is small or $f_1$ is close to $f_2$ (so the effective number of useful events is small).  As $f_1\rightarrow f_2$, the two mixtures become identical and there is thus no way to distinguish $M_1$ and $M_2$; in the context of LLP, this corresponds to attempting to solve a degenerate system of equations.  With sufficiently many training examples and/or well-separated fractions $f_1$ and $f_2$, the techniques trained with $M_1$ and $M_2$ converge to the fully supervised case, as expected from Theorem~\ref{th:optXYAB}.

\begin{figure}[t] 
\centering
\includegraphics[width=0.48\textwidth]{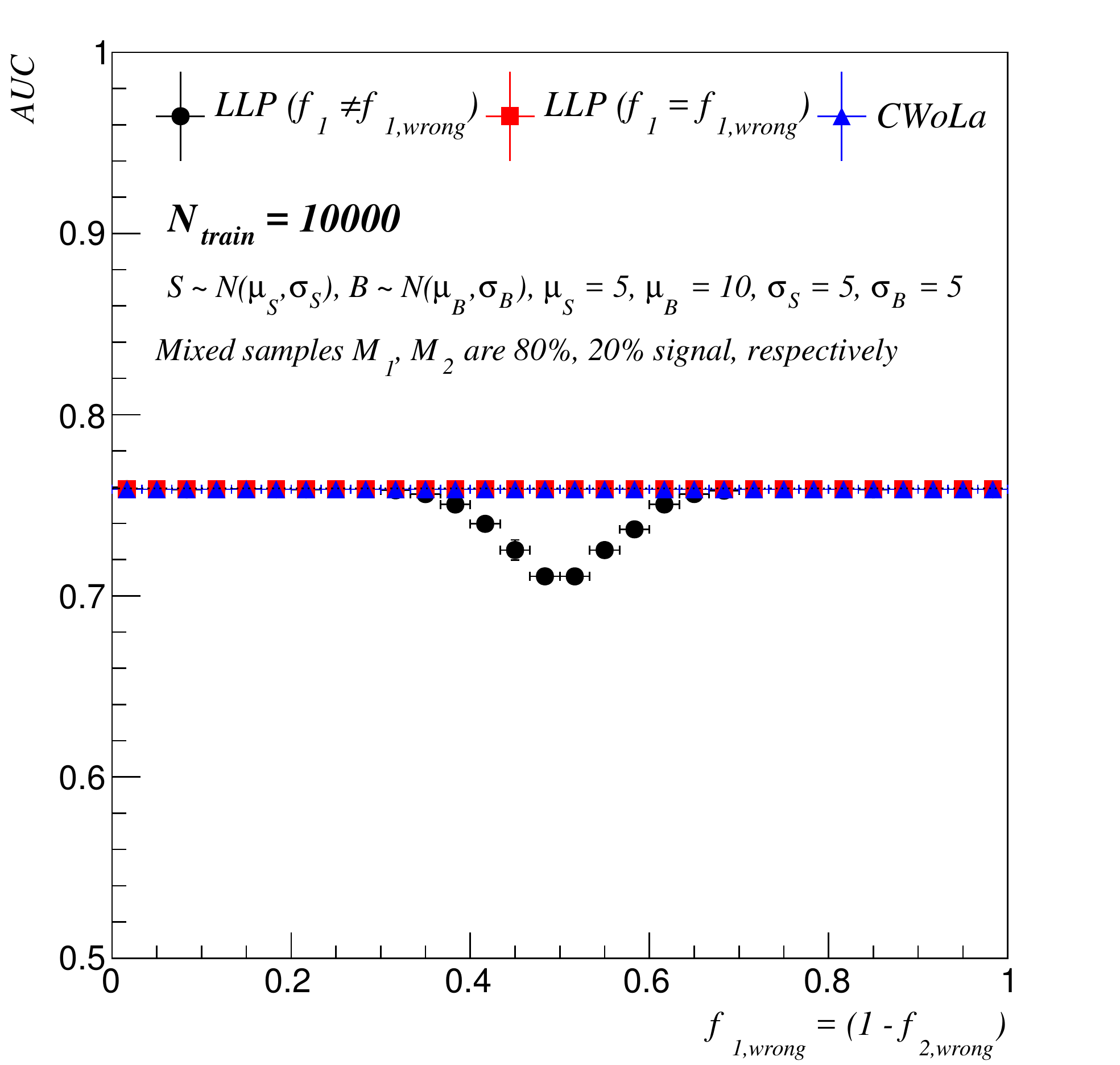}
\caption{The AUC for LLP and CWoLa as a function of the (possibly incorrect) signal fraction provided for training.  By construction, CWoLa does not depend on the input fraction and LLP is only sensitive to provided signal fraction information when that fraction is near $50\%$.}
\label{fig:to2y}
\end{figure}

One advantage of CWoLa over the LLP approach is that the fractions $f_1$ and $f_2$ are not required for training. In \Fig{fig:to2y}, we demonstrate the impact on the AUC for LLP when the wrong fractions are provided at training time.  Here, the true fractions are $f_1=80\%$ and $f_2=20\%$, but different fractions $f_\text{$1$,wrong} = 1- f_\text{$2$,wrong}$ are used to calculate \Eq{eq:weak}.  For $f_\text{$1$,wrong}$ far from $50\%$, there is little dependence on the fraction used for training.  This insensitivity is likely due to the preservation of monotonicity to the full likelihood with small perturbations in $f$, as discussed in detail in \Ref{Cohen:2017exh}.

With this one-dimensional example, the estimate for the optimal classifier under each of the three schemes is computable directly.  It is often the case that $\vec{x}$ is highly multi-dimensional, though, in which case a more sophisticated learning scheme may be required.  We investigate the performance of CWoLa in a five-dimensional space in the next section.

\section{Realistic example: Quark/gluon jet discrimination}
\label{sec:qg}

Quark- versus gluon-initaited jet tagging \cite{Nilles:1980ys,Jones:1988ay,Fodor:1989ir,Jones:1990rz,Lonnblad:1990qp,Pumplin:1991kc,Gallicchio:2011xq,Gallicchio:2012ez,Larkoski:2014pca} is a particularly important classification problem in high energy physics where training on data would be beneficial.  This is because correlations between key observables known to be useful for tagging are not always well-modeled by simulations as they depend on the detailed structure of a jet's radiation pattern~\cite{Aad:2014gea,Aad:2016oit}.  Furthermore, even the LLP paradigm proposed in \Ref{Dery:2017fap} can be sensitive to the input fractions which are themselves dependent on non-perturbative information from parton distribution functions.  In this section, we test the performance of CWoLa in a realistic context where a small number of quark/gluon discriminants are combined into one classifier, similar to the CMS quark/gluon likelihood \cite{CMS:2013kfa,CMS-DP-2016-070}.

A key limitation of this study is that we artificially construct mixed samples $\mathcal{M}_1$ and $\mathcal{M}_2$ from pure ``quark'' ($\mathcal{S}$) and pure ``gluon'' ($\mathcal{B}$) samples.\footnote{The reason for the scare quotes is discussed at length in \Ref{Gras:2017jty}, as the definition of a quark or gluon jet is fundamentally ambiguous.}  In the practical case of interest at the LHC, one would measure a quark-enriched sample in $Z$ plus jet events and a gluon-enriched sample in dijet events, with more sophisticated selections possible as well \cite{Gallicchio:2011xc}.  However, the ``quark'' jet in $pp \to Z+j$ event is not the same as the ``quark'' jet in $pp \to 2j$, since there are soft color correlations with the rest of the event.  Jet grooming techniques \cite{Butterworth:2008iy,Ellis:2009su,Ellis:2009me,Krohn:2009th,Dasgupta:2013ihk,Larkoski:2014wba} can mitigate the impact of soft effects to provide a more universal ``quark'' jet definition \cite{Frye:2016okc,Frye:2016aiz}.  Still, one needs to validate the robustness of quark/gluon classifiers to the possibility of sample-dependent labels, and we leave a detailed study of this effect to future work.

This study is based on five key jet substructure observables which are known to be useful quark/gluon discriminants~\cite{Gras:2017jty}.  The discriminants are combined using a modern neural network employing either CWoLa or fully-supervised learning. We do not show a benchmark curve for LLP since it is difficult to ensure a fair comparison. By contrast, CWoLa and full supervision use the same loss function with the same training strategy, so a direct comparison is meaningful. All of the observables can be written in terms of the generalized angularities~\cite{Larkoski:2014pca} (see also \cite{Berger:2003iw,Almeida:2008yp,Ellis:2010rwa}):
\begin{equation}\label{eq:genang}
\lambda_\beta^\kappa = \sum_{i\in \text{jet}} z_i^\kappa \theta_i^\beta, \quad \text{with} \quad z_i=\frac{p_\text{T,$i$}}{\sum_{j\in \text{jet}} p_\text{T,$j$}}, \quad \theta_i = \frac{\Delta R_i}{R},
\end{equation}
where $\Delta R_i$ is the rapidity/azimuth distance to the $E$-scheme jet axis,\footnote{This is in contrast to \Ref{Gras:2017jty}, which uses the winner-take-all axis \cite{Bertolini:2013iqa,Larkoski:2014uqa,Salam:WTAUnpublished}.} $p_\text{T,$i$}$ is the particle transverse momentum, and $R$ is the jet radius.  The observables used to train the network use $(\kappa, \beta)$ values of:
\begin{equation}
\arraycolsep=5pt
\begin{array}{ccccc}
(0,0) & (2,0) & (1,0.5) & (1,1) & (1,2) \\
\text{multiplicity} &  p_\text{T}^\text{D} &  \text{LHA} & \text{width} & \text{mass}
\end{array}
\end{equation}
where the names map onto the well-known discriminants in the quark/gluon literature.\footnote{Strictly speaking $(2,0)$ is the square of $p_\text{T}^\text{D}$~\cite{Chatrchyan:2012sn}, and $(1,2)$ is mass-squared over energy-squared in the soft-collinear limit.  For this study, we use the angularity definition of the five observables.  Note that the first observable is infrared and collinear (IRC) unsafe, the second observable is IR safe but C unsafe, and the last three observables with $\kappa = 1$ are all IRC safe.  LHA refers to the Les Houches Angularity from the eponymous study in \Refs{Gras:2017jty,Badger:2016bpw}.}

Quark and gluon jets are simulated from the decay of a heavy scalar particle $H$ with $m_H=500$ GeV in either the $pp\rightarrow H\rightarrow q\bar{q}$ or $pp\rightarrow H\rightarrow gg$ channel.  Production, decay, and fragmentation are modeled with \textsc{Pythia 8.183}~\cite{Sjostrand:2007gs}.  Jets are clustered using the anti-$k_t$ algorithm~\cite{Cacciari:2008gp} with radius $R=0.6$ implemented in \textsc{Fastjet 3.1.3}~\cite{Cacciari:2011ma}.  Only detector-stable hadrons are used for jet finding.  Since the gluon color factor $C_A$ is larger than the quark color factor $C_F$ by about a factor of two, gluon jets have more particles and are ``wider'' on average as measured by the angularities listed above.

To classify quarks and gluons with either the CWoLa or fully-supervised method, we use a simple neural network consisting of two dense layers of 30 nodes with rectified linear unit (ReLU) activation functions connected to a 2-node output with a softmax activation function.  All neural network training was performed with the \textsc{Python} deep learning library \textsc{Keras}~\cite{keras} with a \textsc{Tensorflow}~\cite{tensorflow} backend. The data consisted of 200k quark/gluon events, partitioned into 20k validation event, 20k test events, and the remainder used as training event samples of various sizes. He-uniform weight initialization~\cite{heuniform} was used for the model weights. The network was trained with the categorical cross-entropy loss function using the \textsc{Adam} algorithm~\cite{adam} with a learning rate of 0.001 and a batch size of 128.

\begin{figure}[p] 
\centering
\subfloat[]{
\includegraphics[width=0.48\textwidth]{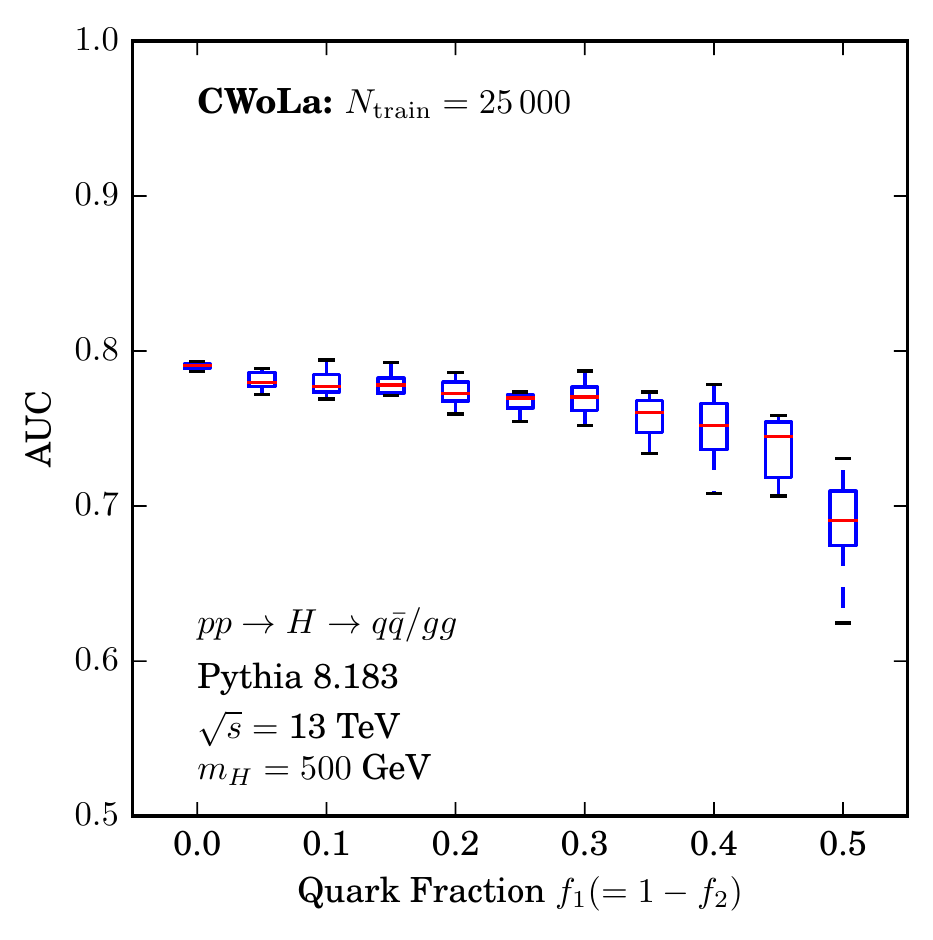}
}
\subfloat[]{
\includegraphics[width=0.48\textwidth]{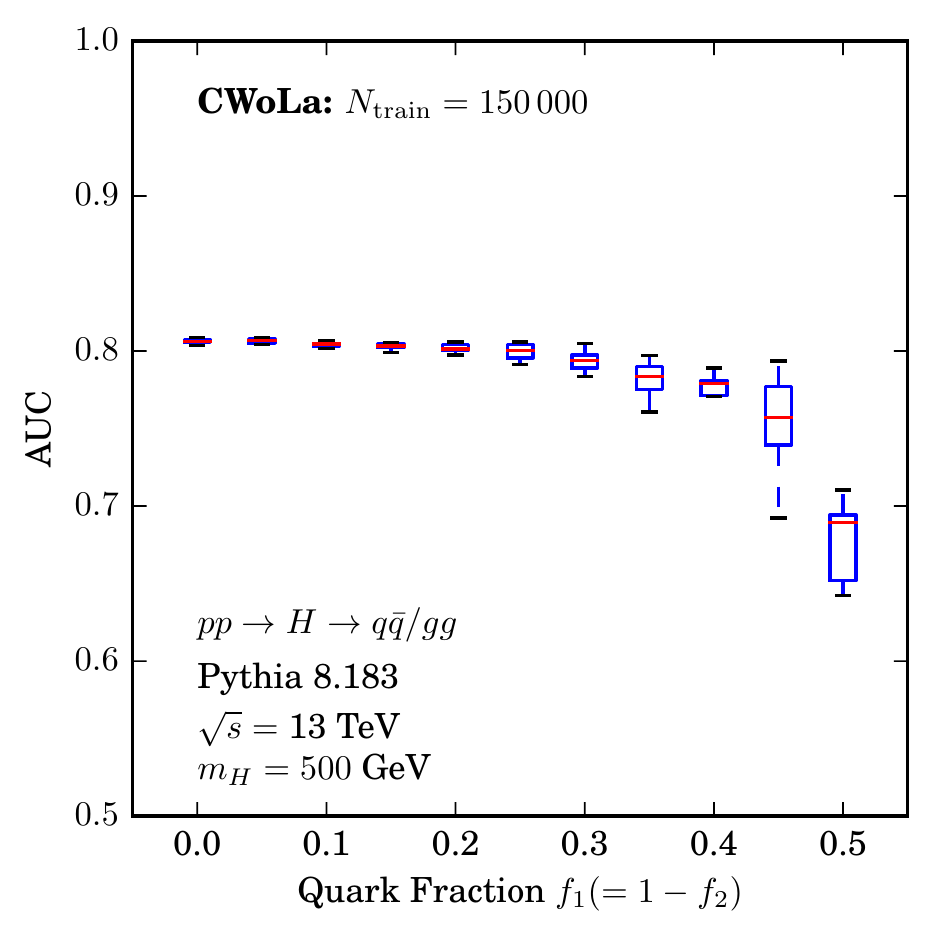}
}
\caption{Training performance of the CWoLa method on two mixed samples with $f_1 = 1-f_2$ quark fraction.  Shown are the range of AUC values obtained from 10 repetitions of training the neural network on (a) 25k events and (b) 150k events for 10 epochs. }
\label{fig:qg150fsweep}
\end{figure}

In \Fig{fig:qg150fsweep}, we show the performance of CWoLa training for quark/gluon classification using mixed samples of different purities. These mixed samples of 25k and 150k training events  were generated by shuffling the pure samples into two sets in different proportions. Performance is measured in terms of the classifier AUC. The behavior resembles that found in the toy model of \Fig{fig:toy}, with more training data resulting in increased robustness to sample impurity.  It is remarkable that such good performance can be obtained even when the signal/background events are so heavily mixed.

\begin{figure}[p] 
\centering
\subfloat[]{
\includegraphics[width=0.48\textwidth]{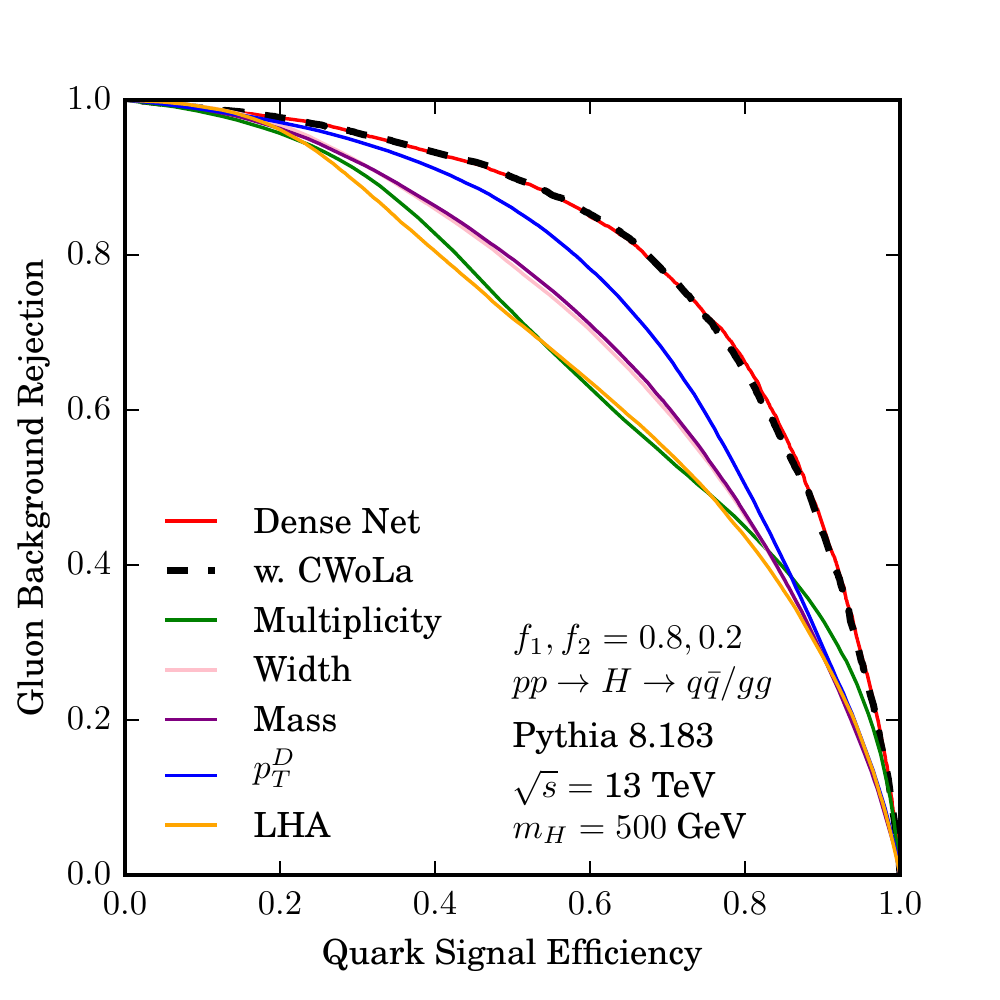}
}
\subfloat[]{
\includegraphics[width=0.48\textwidth]{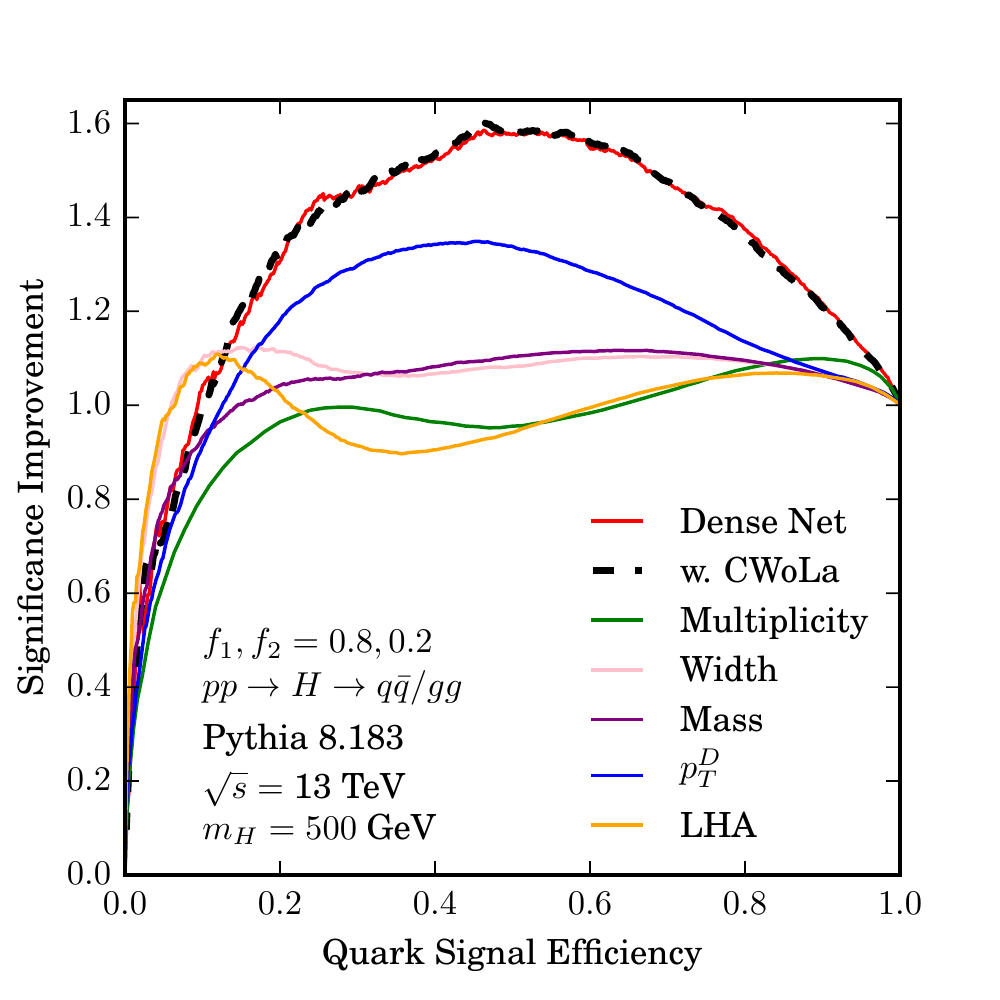}
}
\caption{Quark/gluon discrimination performance in terms of (a) ROC curves and (b) SI curves.  Shown are results for the dense net trained on 150k pure samples, and then with CWoLa on $f_1 = 80\%$ versus $f_2 = 20\%$ mixed samples, as well as the input observables individually.  The classifier trained on the mixed samples achieves similar performance to the classifier trained on the pure samples, with improvement in performance over the input observables.}
\label{fig:qgROCs}
\end{figure}

In \Fig{fig:qgROCs}, we show ROC and significance improvement (SI) curves for 150k training events, where SI is a curve of $\epsilon_q/\sqrt{\epsilon_g}$ at different $\epsilon_q$ values~\cite{Gallicchio:2012ez}.  Results are given for the fully-supervised classifier trained on pure samples and the CWoLa classifier trained on mixed samples with $f_1 = 80\%$ and $f_2 = 20\%$, along with the curves of the input observables.  Both the fully-supervised and CWoLa dense networks achieve similar performance, with the expected improvement over the individual input observables.  This suggests that the proof of CWoLa optimality in Theorem~\ref{th:optXYAB} is achievable in practice, though many more studies are needed to demonstrate this in a wider range of contexts.

\section{Conclusions}
\label{sec:conc}

We introduced the CWoLa framework for training classifiers on different mixed samples of signal and background events, without using true labels or class proportions.  The observation that the optimal classifier for mixed samples of signal and background is also optimal for pure samples of signal and background, proven in Theorem~\ref{th:optXYAB}, could be of tremendous practical use at the LHC for learning directly from data whenever truth information is unknown or uncertain and whenever detailed and reliable simulations are unavailable. We highlight that no new specific code, loss function, or model architecture is needed to implement CWoLa.  Any tools for training a classifier using truth information can be directly applied to discriminate mixed samples and thus to train in the CWoLa framework directly on data.

Using a toy example, we found that CWoLa performs as well as LLP (which requires knowledge of the class proportions), suggesting that CWoLa is a robust paradigm for weak supervision.  Of course, to determine operating points and classification power for the CWoLa method, some label information is needed, but it can be furnished by a smaller sample of testing data that can be separate from the larger mixed samples used for training.  It is also worth remembering that CWoLa assumes that the mixed samples are not subject to contamination or sample-dependent labeling, though one could imagine using data-driven cross-validation with more than two mixed samples to identify and mitigate such effects.  More ambitiously, one could try to apply CWoLa to event samples that otherwise look identical, to try to tease out potential subpopulations of events. 

As a realistic example, we applied the CWoLa framework to the important case of quark/gluon discrimination, a classification task for which simulations are typically unreliable and true labels are unknown.  We showed that the CWoLa method can be successfully used to train a dense neural network for quark/gluon classification on mixed samples with five jet substructure observables as input.  Though the realistic example made use of a neural network, the CWoLa paradigm can be used to train many other types of classifiers. While in this study we considered a relatively small network on a small (but important) number of inputs, the same principles apply for any type of model or input.   In future work, we plan to study CWoLa in the context of deeper architectures and larger inputs.

\acknowledgments

The authors are grateful to Timothy Cohen, Kyle Cranmer, Marat Freytsis, Patrick Komiske, Bryan Ostdiek, Francesco Rubbo, Matthew Schwartz, and Clayton Scott for helpful discussions and suggestions. Cloud computing resources were provided through a Microsoft Azure for Research award. The work of E.M.M. and J.T. is supported by the DOE under grant contract numbers DE-SC-00012567 and DE-SC-00015476.  The work of B.N. is supported by the DOE under contract DE-AC02-05CH11231.

\bibliographystyle{jhep}
\bibliography{myrefs}

\end{document}